\documentclass[11pt]{article}
\usepackage{graphicx}
\usepackage{amsmath}
\usepackage{amssymb}
\usepackage{latexsym}
\usepackage{epsfig,enumerate,amsmath,amsfonts,amssymb,setspace}
\usepackage{indentfirst}
\usepackage{epsfig,float}
\usepackage{wrapfig,lipsum}
\usepackage{mathrsfs}
\usepackage{algorithm}
\usepackage[algo2e,ruled,vlined]{algorithm2e}
\def \nG {n_{_G}}
\def \nH {n_{_H}}

\newcommand{\inR}{{\rm in}}
\newcommand{\outR}{{\rm out}}

\newcommand{\gt}{\gamma_t}
\newcommand{\semiT}{\gamma_{\rm t2}}
\newcommand{\st}{{\rm st}}
\newcommand{\cF}{\mathcal{F}}
\newcommand{\cC}{\mathcal{C}}

\newcommand{\1}{ \vspace{0.1cm} }

\newtheorem{definition}{Definition}[section]
\newtheorem{theorem}{Theorem}[section]
\newtheorem{lemma}{Lemma}[section]
\newtheorem{cor}{Corollary}[section]

\newtheorem{observation}{Observation}[section]
\newtheorem{claim}{Claim}[section]

\newcommand{\qed}{$\Box$}

\let\oldenumerate\enumerate
\renewcommand{\enumerate}{
  \oldenumerate
  \setlength{\itemsep}{0pt}
  \setlength{\parskip}{0pt}
  \setlength{\parsep}{0pt}
}

\begin{document}

\title{Algorithmic Aspects of Semitotal Domination in Graphs}

\author{$^{1}$Michael A. Henning\thanks{Research
supported in part by the University of Johannesburg and the South African National Research Foundation.} and $^{2}$Arti Pandey \\ \\
$^{1}$Department of Pure and Applied Mathematics \\
University of Johannesburg \\
Auckland Park, 2006 South Africa\\
\small \tt Email: mahenning@uj.ac.za
\\
$^{2}$Department of Mathematics\\
Indian Institute of Technology Ropar\\
Nangal Road, Rupnagar, Punjab 140001, INDIA \\
\small \tt Email: arti@iitrpr.ac.in \\
\\
}

\date{}
\maketitle

\begin{abstract}
For a graph $G=(V,E)$, a set $D \subseteq V$ is called a semitotal dominating set of $G$ if $D$ is a dominating set of $G$, and every vertex in $D$ is within distance~$2$ of another vertex of~$D$. The \textsc{Minimum Semitotal Domination} problem is to find a semitotal dominating set of minimum cardinality. Given a graph $G$ and a positive integer $k$, the \textsc{Semitotal Domination Decision} problem is to decide whether $G$ has a semitotal dominating set of cardinality at most $k$. The \textsc{Semitotal Domination Decision} problem is known to be NP-complete for general graphs. In this paper, we show that the \textsc{Semitotal Domination Decision} problem remains NP-complete for planar graphs, split graphs and chordal bipartite graphs. We give a polynomial time algorithm to solve the \textsc{Minimum Semitotal Domination} problem in interval graphs. We show that the \textsc{Minimum Semitotal Domination} problem in a graph with maximum degree~$\Delta$ admits an approximation algorithm that achieves the approximation ratio of $2+3\ln(\Delta+1)$, showing that the problem is in the class log-APX. We also show that the \textsc{Minimum Semitotal Domination} problem cannot be approximated within $(1 - \epsilon)\ln |V| $ for any $\epsilon > 0$ unless NP $\subseteq$  DTIME $(|V|^{O(\log  \log  |V|)})$. Finally, we prove that the \textsc{Minimum Semitotal Domination} problem is APX-complete for bipartite graphs with maximum degree $4$.
\end{abstract}


\newpage
\section{Introduction}

A \emph{dominating set} in a graph $G$ is a set $S$ of vertices of $G$ such that every vertex in $V(G) \setminus S$ is adjacent to at least one vertex in $S$. The \emph{domination number} of $G$, denoted by $\gamma(G)$, is the minimum cardinality of a dominating set of $G$. More thorough treatment of domination, can be found in the books~\cite{hhs1,hhs2}.

A \emph{total dominating set}, abbreviated a TD-set, of a graph $G$ with no isolated vertex is a set $S$ of vertices of $G$ such that every vertex in $V(G)$ is adjacent to at least one vertex in $S$. The \emph{total domination number} of $G$, denoted by $\gt(G)$, is the minimum cardinality of a TD-set of $G$. Total domination is now well studied in graph theory. The literature on the subject of total domination in graphs has been surveyed and detailed in the recent book~\cite{total2}. A survey of total domination in graphs can also be found in~\cite{total1}.

A relaxed form of total domination called semitotal domination was introduced by Goddard, Henning and McPillan~\cite{semi-total1}, and studied further in~\cite{semi-total2,semi-total3,semi-total4,semi-total5} and elsewhere. A set $S$ of vertices in a graph $G$ with no isolated vertices is a \emph{semitotal dominating set}, abbreviated a semi-TD-set, of $G$ if $S$ is a dominating set of $G$ and every vertex in $S$ is within distance~$2$ of another vertex of $S$. The \emph{semitotal domination number} of $G$, denoted by $\semiT(G)$, is the minimum cardinality of a semi-TD-set of $G$. Since every TD-set is a semi-TD-set, and since every semi-TD-set is a dominating set, we have the following observation.

\begin{observation}{\rm (\cite{semi-total1})}
 \label{ob:chain}
For every graph $G$ with no isolated vertex, $\gamma(G) \le \semiT(G) \le \gamma_t(G)$.
\end{observation}

By Observation~\ref{ob:chain}, the semitotal domination number is squeezed between arguably the two most important domination parameters, namely the domination number and the total domination number.

The \textsc{Minimum Domination} problem is to find a dominating set of cardinality $\gamma(G)$. Given a graph $G$ and an integer $k$, the \textsc{Domination Decision} problem is to determine whether $G$ has a dominating set of cardinality at most $k$.  The \textsc{Minimum Total Domination} problem is to find a total dominating set of cardinality $\gt(G)$.  The \textsc{Minimum Semitotal Domination} problem is to find a semi-TD-set of minimum cardinality.
More formally, the minimum semitotal domination problem and its decision version are defined as follows:
\1

\noindent\underline{\textsc{Minimum Semitotal Domination} problem (MSDP)}
\\
[-12pt]
\begin{enumerate}
  \item[] \textbf{Instance}: A graph $G=(V,E)$.
  \item[] \textbf{Solution}: A semi-TD-set $D$ of $G$.
  \item[] \textbf{Measure}: Cardinality of the set $D$.
\end{enumerate}

\noindent\underline{\textsc{Semitotal Domination Decision} problem (SDDP)}
\\
[-12pt]
\begin{enumerate}
  \item[] \textbf{Instance}: A graph $G=(V,E)$ and a positive integer $k \le |V|$.
  \item[] \textbf{Question}: Does there exist a semi-TD-set $D$ in $G$ such that $|D| \le k$?
\end{enumerate}

The \textsc{Semitotal Domination Decision} problem is known to be NP-complete for general graphs~\cite{semi-total2}. On the positive side, a linear time algorithm exists to find a minimum semi-TD-set in trees~\cite{semi-total2}. In this paper, we further continue the algorithmic study of the \textsc{Minimum Semitotal Domination} problem. The main contributions of the paper are summarized below. In Section~\ref{notation}, we discuss some pertinent definitions. In Section~\ref{S:prelim}, we present some preliminary results and discuss some complexity difference between total domination and semitotal domination problem. In Section~\ref{S:NPresults}, we prove that the \textsc{Semitotal Domination Decision} problem remains NP-complete for chordal bipartite graphs, planar graphs and split graphs. In Section~\ref{S:interval}, we present a polynomial time algorithm to compute a minimum cardinality semi-TD-set of interval graphs, an important subclass of chordal graphs. In Section~\ref{S:apx}, we propose an approximation algorithm for the problem. In this section, we also discuss some approximation hardness results. Finally, Section~\ref{S:conclude}, concludes the paper.

\section{Terminology and Notation}
\label{notation}

For notation and graph theory terminology, we in general follow~\cite{total2}. Specifically, let $G = (V, E)$ be a graph with vertex set $V=V(G)$ and edge set $E=E(G)$, and let $v$ be a vertex in $V$. The \emph{open neighborhood} of $v$ is the set $N_G(v) = \{u \in V \, | \, uv \in E\}$ and the \emph{closed neighborhood of $v$} is $N_G[v] = \{v\} \cup N_G(v)$. Thus, a set $D$ of vertices in $G$ is a dominating set of $G$ if $N_G(v) \cap D \ne \emptyset$ for every vertex $v \in V \setminus D$, while $D$ is a total dominating set of $G$ if $N_G(v) \cap D \ne \emptyset$ for every vertex $v \in V$. The \emph{distance} between two vertices $u$ and $v$ in a connected graph $G$, denoted by $d_G(u,v)$, is the length of a shortest $(u,v)$-path in $G$. The \emph{distance} $d_G(v,S)$ between a vertex $v$ and a set $S$ of vertices in a graph $G$ is the minimum distance from~$v$ to a vertex of $S$ in $G$.

For a vertex $v$ in $G$ and an integer~$i \ge 1$, let $N_i(v;G)$ denote the set of all vertices at distance exactly~$i$ from $v$ in $G$. In particular, $N_1(v;G)$ is the open neighborhood, $N_G(v)$, of $v$. Further, let $N_i[v;G]$ denote the set of all vertices within distance~$i$ from $v$ in~$G$. If the graph $G$ is clear from the context, we omit it in the above expressions. For example, we write $N(v)$, $N[v]$ and $d(u,v)$ rather than $N_G(v)$, $N_G[v]$ and $d_G(u,v)$, respectively.

For a set $S \subseteq V(G)$, the subgraph induced by $S$ is denoted by $G[S]$.  If $G[C]$, where $C \subseteq V$, is a complete subgraph of $G$, then $C$ is a \emph{clique} of $G$. A set $S \subseteq V$ is an \emph{independent set} if $G[S]$ has no edge. A graph $G$ is \emph{chordal} if every cycle in $G$ of length at least four has a \emph{chord}, that is, an edge joining two non-consecutive vertices of the cycle. A chordal graph $G = (V, E)$ is a \emph{split graph} if $V$ can be partitioned into two sets $I$ and $C$ such that $C$ is a clique and $I$ is an independent set.

A graph $G = (V,E)$ is \emph{bipartite} if $V$ can be partitioned into two disjoint sets $X$  and $Y$ such that every edge of $G$ joins a vertex in $X$ to a vertex  in $Y$, and such a partition $(X,Y)$ of $V(G)$ is called a \emph{bipartition} of $G$. Further, we denote such a bipartite graph $G$ by $G=(X,Y,E)$. A bipartite graph $G$ is a \emph{chordal bipartite} if every cycle of length at least~$6$ has a chord. A graph $G$ is a \emph{planar graph} if it can be drawn on the plane in such a way that no two edges cross each other except at a vertex. Such a drawing is called a \emph{planar embedding} of the planar graph.

A graph $G$ is an \emph{interval graph} if there exists a one-to-one correspondence between its vertex set and a family of closed intervals in the real line, such that two vertices are adjacent if and only if their corresponding intervals intersect. Such a family of intervals is called an \emph{interval model} of a graph.

Let $D = (V,A)$ be a digraph with vertex set $V$ and arc set $A$, and let $v$ be a vertex of~$D$. We write $(u,v) \in A$ to denote an arc directed from $u$ to $v$. By a \emph{path} in $D$, we mean a directed path.

In the rest of the paper, all graphs considered are simple connected graphs with at least two vertices, unless otherwise mentioned specifically. We use the standard notation $[k] = \{1,\ldots,k\}$.
For most of the approximation related terminologies, we refer to  \cite{ausiello}.

\section{Preliminary Result}
\label{S:prelim}

In this subsection, we make some complexity difference between total domination and semitotal domination problem. The \textsc{Minimum Total Domination} problem is polynomial time solvable for chordal bipartite graphs~\cite{total-cbg}, but in Section~\ref{S:NPresults}, we will show that the \textsc{Semitotal Domination Decision} problem is NP-complete for this graph class.

On the other hand, we define a graph class, called GP$4$-\emph{graph}, for which the decision version of the total domination problem is NP-complete, but the \textsc{Minimum Semitotal Domination} problem is easily solvable.

\begin{definition}{\rm (GP$4$-graph)}
\label{defn1}
A graph $G=(V,E)$ is a \emph{GP}$4$-\emph{graph} if it can be obtained from a general connected graph $H=(V_{H},E_{H})$ where $V_{H}=\{v_{1},v_{2},\ldots,v_{\nH}\}$, by adding a path of length~$4$ to every vertex of $H$ so that the resulting paths are vertex disjoint. Formally, $V = V_{H} \cup \{w_i, x_{i},y_{i},z_{i} \mid i \in [\nH] \, \}$ and $E=E_{H}\cup \{v_{i}w_{i},w_{i}x_{i},x_{i}y_{i},y_{i}z_{i}\mid i \in [\nH] \, \}$.
\end{definition}

Let $G$ be a GP$4$-graph of order~$\nG = 5\nH$ as constructed in Definition~\ref{defn1}. Let $V_i = \{v_i,w_i, x_i, y_i, z_i\}$ for $i \in [\nH]$. If $S$ is a semi-TD-set of $G$, then the set $S$ contains at least two vertices from each set $V_i$ in order to dominate the vertices in $V_i \setminus \{v_i\}$ for each $i \in [\nH]$. Thus, $\semiT(G) \ge 2\nH$. However, the set $\{w_i,y_i \mid i \in [\nH]\}$ is a semi-TD-set of $G$, and so $\semiT(G) \le 2\nH$. Consequently, $\semiT(G) = 2\nH = 2\nG/5$. We state this formally as follows.

\begin{observation}
 \label{ob:GP4}
If $G$ is a \emph{GP}$4$-\emph{graph}, then $\semiT(G) = \frac{2}{5}|V(G)|$.
\end{observation}

\begin{lemma}
\label{l:lemGP4}
If $G$ is a GP$4$-graph constructed from a graph $H$ as in Definition~\ref{defn1}, then $H$ has a TD-set of cardinality~$k$ if and only if $G$ has a TD-set of cardinality~$2\nH +k$.
\end{lemma}
\begin{proof}
Let $D$ be a TD-set of $H$ and $|D|=k$. Then, $D \cup \{w_{i},y_{i} \mid i \in [ \nH ] \}$ is a TD-set of $G$ of cardinality $2\nH +k$.

Conversely, suppose that $D'$ is a TD-set of $G$ of cardinality $2\nH+k$. In order to totally dominate the vertex $z_i$, we note that $y_i \in D'$ for all $i \in [\nH]$. Further, to totally dominate the vertex $y_{i}$, we note that $x_{i}$ or $z_{i}$ must belong to $D'$. Now define $D = D' \setminus \{x_{i},y_{i},z_{i} \mid i \in [\nH] \}$. Then, $|D|\le k$. Also, $D$ totally dominates all the vertices of $V_{H}$. Note that a vertex $v_{i} \in V_{H}$ is totally dominated by either $w_{i}$ or by one of the neighbors of $v_i$ in $H$. If $w_{i} \in D$ for some $i \in [\nH]$, then we simply replace the vertex $w_i$ in $D$ by an arbitrary neighbors of $v_i$ in $H$. Hence, we can choose the set $D$ so that $D$ is a subset of $V_{H}$, implying that $D$ is a TD-set of $H$ of cardinality at most~$k$. This proves that $H$ has a TD-set of cardinality~$k$.~\qed
\end{proof}

%\medskip
Since the decision version of the \textsc{Minimum Total Domination} problem is already known to be NP-complete for general graphs~\cite{total-npc}, the following theorem follows directly from Lemma~\ref{l:lemGP4}.

\begin{theorem}
The decision version of the \textsc{Minimum Total Domination} problem is NP-complete for GP$4$-graphs.
\end{theorem}

\section{NP-Completeness Results}
\label{S:NPresults}

In this section, we study the NP-completeness of the \textsc{Semitotal Domination Decision} problem. The \textsc{Semitotal Domination Decision} problem is NP-complete for general graphs. We strengthen the complexity result of the \textsc{Semitotal Domination Decision} problem, by showing that it remains NP-complete for planar graphs, chordal bipartite graphs and split graphs.

\subsection{Results for Chordal Bipartite Graphs and Planar Graphs}

In this section, we prove the hardness result for the \textsc{Semitotal Domination Decision} problem in chordal bipartite graphs and planar graphs. The proof involves a reduction from the \textsc{Domination Decision} problem. The following NP-completeness result is already known for the  \textsc{Domination Decision} problem.

\begin{theorem}{\rm (\cite{np,muller})}
\label{t:known1}
The \textsc{Domination Decision} problem is NP-complete for bipartite graphs, chordal graphs and planar graphs. It also remains NP-complete for chordal bipartite graphs (a subclass of bipartite graphs) and split graphs (a subclass of chordal graphs).
\end{theorem}

\begin{theorem}
\label{t:bipartite}
The \textsc{Semitotal Domination Decision} problem is NP-complete for bipartite graphs.
\end{theorem}
\begin{proof}
Clearly, the \textsc{Semitotal Domination Decision} problem is in NP. To show the hardness, we give a polynomial reduction from the \textsc{Minimum Domination} problem. Given a non-trivial bipartite graph $G=(V,E)$, where $V=\{v_{1},v_{2},\ldots,v_{n}\}$, we construct a graph $H=(V_{H},E_{H})$ as follows: For each $i \in [n]$, add the path $x_{i}y_{i}z_{i}u_{i}w_{i}$ on five vertices and the edge $v_{i}z_{i}$ in $H$.

\begin{figure}
 \begin{center}
  \includegraphics[width=9cm, height=5cm]{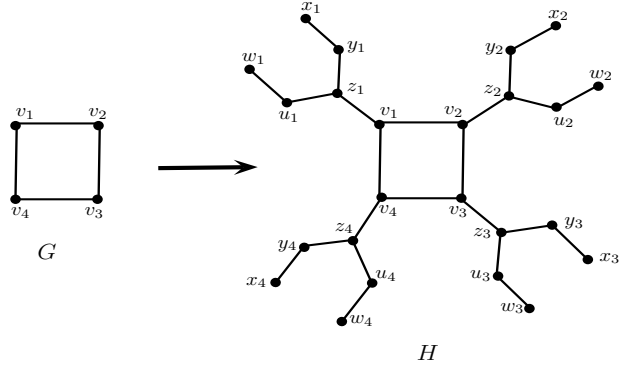}\\
 \caption{An illustration of the construction of $H$ from $G$  in the proof of Theorem~\ref{t:bipartite}.}
\label{f:bipartite}
\end{center}
\end{figure}

Formally, the vertex set $V_{H} = V(G) \cup \{x_{i},y_{i},z_{i},u_{i},w_{i}\mid i \in [n] \}$, and the edge set $E_{H} = E(G) \cup \{x_{i}y_{i},y_{i}z_{i},z_{i}u_{i},u_{i}w_{i},v_{i}z_{i}\}$. We note that since $G$ is a bipartite graph, so too is the graph~$H$. In the special case when $G$ is a $4$-cycle, the construction of the graph $H$ from the bipartite graph $G$ is illustrated in Fig.~\ref{f:bipartite}. Now to complete the proof, it suffices for us to prove the following claim:

\begin{claim}
\label{c:claim1}
The graph $G$ has a dominating set of cardinality at most~$k$ if and only if $H$ has a semi-TD-set of cardinality at most~$2n+k$.
\end{claim}
\begin{proof}
Let $D$ be a dominating set of $G$ of cardinality at most~$k$, and consider the set $D_H = D \cup \{u_{i},y_{i}\mid i \in [n]\}$. Since $D$ is a dominating set of $G$, and since $\{u_i,y_i\}$ dominates the five vertices on the added path $x_{i}y_{i}z_{i}u_{i}w_{i}$, the set $D_H$ is a dominating set of $H$. We note that $d_H(u_i,y_i) = 2$ for all $i \in [n]$. Further, if $v$ is an arbitrary vertex of $D$, then $v = v_i$ for some $i \in [n]$, and the vertex $v$ is at distance~$2$ from both $u_i$ and $y_i$ in $H$. Thus, $D_H$ is a semi-TD-set of $H$ of cardinality~$2n + |D| \le  2n+k$.

Conversely, suppose that $D'$ is a semi-TD-set  of $H$ of cardinality at most~$2n+k$. Since every semi-TD-set is also dominating set, in order to dominate the vertex $x_i$, we note that $x_{i}$ or $z_{i}$ must belong to $D'$ for each $i \in [n]$. Similarly, in order to dominate the vertex $w_i$, we note that $u_{i}$ or $w_{i}$ must belong to $D'$ for each $i \in [n]$.
We now define $D = D' \setminus \{x_{i},y_{i},u_{i},w_i \mid i \in [n] \}$. By our earlier observations, $|D| \le k$. Also, the set $D$ dominates all the vertices of $V(G)$. Moreover, a vertex $v_{i} \in V(G)$ is dominated by $z_{i}$ or a vertex from $N_{G}[v_{i}]$ for all $i \in [n]$. If $z_{i} \in D$ for some $i \in [n]$, then we simply replace the vertex $z_i$ in $D$ by the vertex $v_i$. Hence, we can choose the set $D$ so that $D$ is a subset of $V(G)$, implying that $D$ is a dominating set of $G$ of cardinality at most~$k$. This completes the proof of the Claim~\ref{c:claim1}.
\qed
\end{proof}

The proof of Theorem~\ref{t:bipartite} now follows from Theorem~\ref{t:known1} and Claim~\ref{c:claim1}.
\qed
\end{proof}

\medskip
We remark that in Theorem~\ref{t:bipartite}, if $G$ is planar graph (chordal bipartite graph), then the constructed graph $H$ is also planar (chordal bipartite). Hence, Theorem~\ref{t:known1} and Theorem~\ref{t:bipartite} imply the following result.

\begin{theorem}
The \textsc{Semitotal Domination Decision} problem is NP-complete for planar graphs and chordal bipartite graphs.
\end{theorem}

\subsection{Results for Split Graphs}

In this section, we prove the hardness result for the \textsc{Semitotal Domination Decision} problem in split graphs.

\begin{theorem}
\label{t:split}
The \textsc{Semitotal Domination Decision} problem is NP-complete for split graphs.
\end{theorem}
\begin{proof}
Clearly, the \textsc{Semitotal Domination Decision} problem is in NP. To show the hardness, we give a polynomial reduction from the \textsc{Minimum Domination} problem in split graphs. Given a non-trivial split graph $G=(V,E)$ with split partition $(K,I)$, where $K$ is a clique and $I$ an independent set, we construct a graph $H=(V_{H},E_{H})$ as follows:

Let $K=\{v_{1},v_{2},\ldots,v_{p}\}$ and $I=\{u_{1},u_{2},\ldots,u_{q}\}$ and $p+q=n$. Let $X = \{x_1,x_2,\ldots,x_p\}$ and $Y = \{y_1,y_2,\ldots,y_q\}$ be two vertex disjoint sets such that $(X \cup Y) \cap V = \emptyset$. Further, let $W = \{w,z,r,s,t\}$ be a set of five new vertices that do not belong to $X \cup Y \cup V$. We now define the graph $H$ as follows. Let $V_{H} = V \cup W \cup X \cup Y$. Further, let $K_{H} = K \cup Y \cup \{s,w\}$ and $I_{H} = X \cup I \cup \{r,t,z\}$. Let $E_{K}$ denotes the set of edges required to make $H[K_{H}]$ a complete graph. We now define $E_{H} = E \cup E_{K} \cup \{v_{i}x_{i},x_{i}w \mid i \in [p] \} \cup \{u_{j}y_{j},y_{j}t\mid j \in [q] \} \cup\{rs,st,wz\}$. We note that $K_{H}$ is a clique in $H$ and $I_{H}$ is an independent set in $H$. Hence, the graph $H$ is also a split graph with split partition $(K_{H},I_{H})$. The construction of a graph $H$ from a split graph $G$ is illustrated in Fig.~\ref{f:split}. To complete the proof, it suffices for us to prove the following claim.

\begin{figure}
 \begin{center}
  \includegraphics[width=13cm, height=4.5cm]{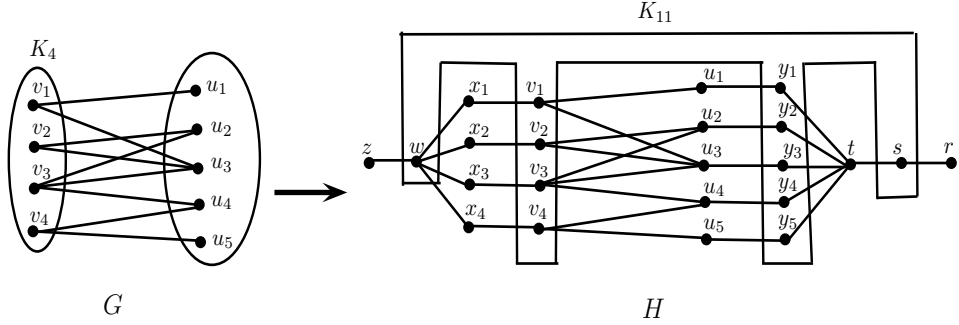}\\
 \caption{An illustration of the construction of $H$ from $G$ in the proof of Theorem~\ref{t:split}.}
\label{f:split}
\end{center}
\end{figure}

\begin{claim}
\label{c:claim2}
The graph $G$ has a dominating set of cardinality at most~$k$ if and only if $H$ has a semi-TD-set of cardinality at most~$k+2$.
\end{claim}
\begin{proof}
Every dominating set of $G$ can be extended to a semi-TD-set of $H$ by adding to it the two vertices $w$ and $s$. Thus, if $G$ has a dominating set of cardinality at most~$k$, then $H$ has a semi-TD-set of cardinality at most~$k+2$.

Conversely, let $D'$ be a semi-TD-set of cardinality at most~$k+2$ in $H$. In order to dominate the vertex $z$, we note that $z$ or $w$ must belong to $D'$. If $z \in D$, then we replace $z$ in $D'$ with the vertex $w$. Hence, we can choose the set $D'$ so that $z \in D'$. Analogously, we can choose the set $D'$ so that $s \in D'$. If $x_i \in D'$ for some $i \in [p]$, then we replace $x_i$ in $D'$ with the vertex $v_i$. Analogously, if $y_j \in D'$ for some $j \in [q]$, then we replace $y_j$ in $D'$ with the vertex~$u_j$. Hence, we can further choose the set $D'$ so that $D' \cap (X \cup Y) = \emptyset$. We now define $D = D' \setminus \{s,r,w,z\}$. By our choice of the set $D$, we note that $D \subseteq V = K \cup I$. Further, $|D| = |D'| - |\{w,s\}| \le k$ and the set $D$ dominates the set $I$. If there is a vertex $v_i$ not dominated by the set $D$ for some $i \in [p]$, then no neighbor of $v_i$ that belongs to the set $I$ would be dominated by $D$, a contradiction. Hence, the set $D$ dominated the set $K$, implying that $D$ is a dominating set of $G$ of cardinality at most~$k$. Thus, if $H$ has a semi-TD-set of cardinality at most~$k+2$, then $G$ has a dominating set of cardinality at most~$k$. This completes the proof of the Claim~\ref{c:claim2}.~\qed
\end{proof}

The proof of Theorem~\ref{t:split} now follows from Theorem~\ref{t:known1} and Claim~\ref{c:claim2}.
\qed
\end{proof}

\section{Algorithm for Interval Graphs}
\label{S:interval}

In this section, we present a polynomial time algorithm to compute a minimum semi-TD-set in an interval graph by reducing it to a shortest path problem in an acyclic directed weighted graph.

A linear time recognition algorithm exists for interval graphs, and for an interval graph an interval family can also be constructed in linear time~\cite{booth,golumbic}. Without loss of generality, we may assume that no two intervals share a common end point. In this paper, we denote by $G_I$ an interval graph associated with a collection of intervals $I$, and directly deal with intervals instead of vertices.

\begin{definition}
\label{d:defnGI}
Let $G_{I}$ be an interval graph associated with an interval model $I=\{I_{1},I_{2},\ldots,I_{n}\}$ of $G_I$, where $I_{i} = [a_{i},b_{i}]$ for $i \in [n]$. A semi-TD-set for $G_{I}$ corresponds to a subset $S$ of intervals in $I$ such that every interval not in $S$ overlaps with some interval in $S$, and for every interval $I_{p} \in S$, one of the following conditions is satisfied:
\\[-24pt]
\begin{enumerate}
  \item[\textbf{C1}:] There exists an interval $I_{q} \in S$ such that $I_{p}$ and $I_{q}$ overlap.
  \item[\textbf{C2}:] There exists a pair of intervals $I_{q}\in S$ and $I_{r}\in I$ such that $I_{r}$ intersects both $I_{p}$ and $I_{q}$.
\end{enumerate}

\end{definition}
We proceed further with the following lemma.

\begin{lemma}
\label{l:interval}
If an interval graph $G_{I}$ has no interval properly containing all the others, then there exists a subset $S$ of intervals of $I$ corresponding to a minimum semi-TD-set of $G_I$ such that no interval in $S$ is properly contained within any other interval in~$I$.
\end{lemma}
\begin{proof}
Let $G_{I}$ be an interval graph which has no interval properly containing all the others. Among all subset of intervals of $I$, let $S$ be chosen so that the following holds.
\\[-22pt]
\begin{enumerate}
\item[(1)] The corresponding set, $D_S$, of vertices of $G_I$ is a minimum semi-TD-set of $G$.
\item[(2)] Subject to (1), the number of intervals in $S$ that are properly contained within other intervals in $I$ is a minimum.
\end{enumerate}

We show that no interval in $S$ is properly contained within any other interval in $I$. Suppose, to the contrary, that there is some interval $I_{k} \in S$ that is properly contained within some other interval in $I$. Among all such intervals containing $I_k$, let $I_r \in I$ be chosen to be of maximum length. We note that $I_k \subset I_r$.

Suppose firstly that $I_{r} \in I \setminus S$. We now consider the set $S' = (S \setminus \{I_k\}) \cup \{I_r\}$. We note that the corresponding set, $D_{S'}$, of vertices of $G_I$ is a minimum semi-TD-set of $G$. By the maximality of the interval $I_r$, we note that $I_r$ is not properly contained within some other interval in $I$. Hence, the number of intervals in $S'$ that are properly contained within other intervals in $I$ is one less than the number of intervals in $S$ that are properly contained within other intervals in $I$, contradicting our choice of the set $S$. Thus, $I_{r} \in S$.

If $I_{r}$ intersects with any other interval of $S$ other than $I_{k}$, then the set $S^* = S \setminus \{I_k\}$ is a subset of intervals in $I$ such that the set, $D_{S^*}$, of vertices of $G_I$ associated with the intervals $S'$ is a semi-TD-set of $G_{I}$ satisfying $|D_{S^*}| < |D_S|$, a contradiction. Hence, the interval $I_k$ is the only interval in $S$ that intersects with the interval $I_r$. On the other hand, since $G_{I}$ is connected, there must exist an interval $I_{h} \in I \setminus S$ such that $I_{h}$ intersects with $I_{r}$. If more than one such intervals exist, we choose the interval $I_{h}$ to be the largest such interval. We now consider the set $S'' = (S \setminus \{I_{k}\}) \cup \{I_{h}\}$. We note that the corresponding set, $D_{S''}$, of vertices of $G_I$ is a minimum semi-TD-set of $G$. By the maximality of the interval $I_r$, we note that $I_r$ is not properly contained within some other interval in $I$. Further, by the maximality of the interval $I_h$ and the observation that $I_k$ is the only interval in $S$ that intersects with the interval $I_r$, we note that $I_h$ is not properly contained within some other interval in $I$. Hence, the number of intervals in $S''$ that are properly contained within other intervals in $I$ is one less than the number of intervals in $S$ that are properly contained within other intervals in $I$, contradicting our choice of the set $S$. We deduce, therefore, that no interval in $S$ is properly contained within any other interval in $I$. This completes the proof of Lemma~\ref{l:interval}.
\qed
\end{proof}

\medskip
Note that if the interval graph $G_{I}$ has an interval, say $I_{k}$, which properly contains all the other intervals, then $I_{k}$ along with any other interval of $G_{I}$ forms a subset of two intervals that corresponds to a minimum semi-TD-set of $G_{I}$. We show next how to find a minimum semi-TD-set of an interval graph $G_I$ which satisfies the condition in the statement of Lemma~\ref{l:interval}.

Adopting our earlier notation, let $G_{I}$ be an interval graph associated with an interval model $I=\{I_{1},I_{2},\ldots,I_{n}\}$ of $G_{I}$, where $I_{i} = [a_{i},b_{i}]$ for $i \in [n]$. Let $V(G_I) =  \{1,2,\ldots,n\}$, where vertex $i$ corresponds to the interval $I_i$ for $i \in [n]$. We now define $I' = I \cup \{I_{0},I_{n+1}\}$ where $I_{0} = [a_{0},b_{0}]$, $I_{n+1}=[a_{n+1},b_{n+1}]$, $b_{0}< \min\{a_{k} \mid k \in [n] \}$ and $a_{n+1} > \max \{b_{k}\mid k \in [n]\}$. We also assume that the intervals in $I'$ are in increasing order of their left end points; that is, $a_{0} < a_{1} < \cdots <a_{n} <a_{n+1}$.

We now construct a directed graph $D = (V,A)$ with vertex set $V$ and arc set $A$ from $I'$ as follows. The vertices in $V$ correspond to the intervals in $I'$ which are not properly contained within other intervals. Thus, $V = \{k \mid I_{k} \in I'$ and $I_{k}$ is not contained in any other interval in $I'\}$. The arcs in $A$ are partitioned into two sets $A_1$ and $A_2$ (that is, $A = A_1 \cup A_2$ and $A_1 \cap A_2 =\emptyset$) as follows. If $I_i$ and $I_j$ are two overlapping intervals in $I'$ where $1 \le i < j \le n$, then we add the arc $(i,j)$ from vertex $i$ to vertex $j$ to the set $A_1$. We note that if $(i,j) \in A_1$, then $a_{i} < a_{j} <b_{i} < b_{j}$. We next define the set of arcs in $A_2$. Suppose that $I_i$ and $I_j$ are two non-overlapping intervals in $I'$ where $0 \le i < j \le n+1$, and so $b_i < a_j$. If there is no interval $I_h$ such that $b_i < a_h$ and $b_h < a_j$, then we add the arc $(i,j)$ from vertex $i$ to vertex $j$ to the set $A_2$. By construction, we note that the directed graph $D$ is acyclic.

We further partition the arcs of $A_2$ in two classes which we label as \emph{marked} and \emph{unmarked}. If $(i,j)$ is an arc in $A_2$ and there exists an interval $I_{h} \in I'$ that overlaps both intervals $I_{i}$ and $I_{j}$, then we call the arc $(i,j)$ a \emph{marked arc}, otherwise we call it an \emph{unmarked arc}. We are now in a position to state the following theorem.

\begin{theorem}
\label{t:interval}
Let $G_{I}$ be an interval graph associated with an interval model $I$ of $G$, and let the set $I'$ of intervals and the digraph $D$ be defined as before. Then any semi-TD-set of $G_I$ such that the associated set of intervals in $I$ has no interval properly contained within any interval in $I$ corresponds to a (directed) path between vertex~$0$ and vertex~$n+1$ in the digraph $D$ which does not include any two consecutive unmarked arcs of $D$.
\end{theorem}
\begin{proof}
Let $P \colon 0,i_{1},i_{2},\ldots,i_{k},n+1$ be a (directed) path from vertex $0$ to vertex $n+1$ in $D$ which does not contain two consecutive unmarked arcs of $D$. We assume that $i_{0}=0$ and $i_{k+1}=n+1$. Define $S= \{I_{i} \mid$ vertex $i$ appears in path $P$ and $i \notin\{0,n+1\}\}$, that is, $S=\{I_{i_{1}},I_{i_{2}},\ldots,I_{i_{k}}\}$, and define $D_{S}=\{i_{1},i_{2},\ldots,i_{k}\}$. Next, we show that $D_{S}$ is a semi-TD-set of $G_{I}$.

Let $(i,j)$ be an arbitrary arc in $P$ where $i \ne 0$ and $j \neq n+1$. If there exists an interval $I_{\ell} = [a_{\ell},b_{\ell}]$ with $i<\ell<j$ which intersects neither $I_{i}$ nor $I_{j}$, then $b_{i} < a_{\ell} < b_{\ell} < a_{j}$, which contradicts the fact that $(i,j)$ is an arc in $D$. Hence, every interval $I_{\ell} = [a_{\ell},b_{\ell}]$ with $i<\ell<j$ intersects at least one of the intervals $I_{i}$ and $I_{j}$. Also, since $(0,i_{1})$ is an arc in $D$, every interval $I_{\ell} = [a_{\ell},b_{\ell}]$ with $1 \le \ell \le i_1 - 1$ intersects the interval $I_{0}$ or $I_{i_{1}}$.  However, the interval $I_{0}$ does not intersect any interval, implying that such an interval $I_\ell$ is intersected by $I_{i_{1}}$. Analogously, the interval $I_{i_{k}}$ intersects the interval $I_{\ell} = [a_{\ell},b_{\ell}]$ for all $\ell$ where $i_{k}+1 \le \ell \le n$. This proves that every interval in $I \setminus S$ is intersected by at least one interval in $S$.

Now, consider an arbitrary interval $I_{r} \in S$, and the corresponding vertex $i_{r}$ in $P$. So, $i_{r}$ will have one incoming arc $(i_{r-1},i_{r})$ and one outgoing arc $(i_{r},i_{r+1})$. Since these are two consecutive arcs of $P$, both of them cannot be unmarked arcs of $A_{2}$. Hence at least one of them must be either an arc of $A_{1}$ or a marked arc of $A_{2}$. If $(i_{r},i_{r+1})$ is an arc of $A_{1}$ or a marked arc of $A_{2}$, then the intervals $I_{r}$ and $I_{r+1}$ either intersect or there exists an interval $I_{h}$ such that $I_{h}$ intersects both intervals $I_{r}$ and $I_{r+1}$. If $(i_{r-1},i_{r})$ is an arc of $A_{1}$ or a marked arc of $A_{2}$, then the interval $I_{r-1}$ and $I_{r}$ either intersects or there exists an interval $I_{h}$ such that $I_{h}$ intersects both intervals $I_{r-1}$ and $I_{r}$. Hence, corresponding to every interval $I_{r}$ in $S$, the conditions \textbf{C1} and \textbf{C2} of Definition~\ref{d:defnGI} holds. Hence, $D_{S}$ is a semi-TD-set of $G_{I}$.

Conversely, let $D_S$ be a semi-TD-set of $G_I$, and $S$ be the set of intervals corresponding to $D_S$. Also assume that no interval in $S$ is properly contained within any other interval in~$I$. Adopting our earlier notation, recall that the intervals in $I'$ are in increasing order of their left end points; that is, $a_{0} < a_{1} < \cdots < a_{n} <a_{n+1}$. Let $D_S = \{i_1,i_2,\ldots,i_k\}$, where $i_1 < i_2 < \cdots < i_k$. Next, we show that there is a (directed) path, $P \colon 0,i_1,i_2,\ldots,i_k,n+1$ from vertex~$0$ to vertex~$n+1$ in the digraph $D$, and this path does not include any two consecutive unmarked arcs of $D$.

Since $b_0 < a_1$, we note that every interval $I_h$ where $h \ge 1$ satisfies $b_0 < a_h$. If there is an interval $I_h$ such that $h \ge 1$ and $b_h < a_{i_1}$, then the interval $I_h$ is not intersected by any interval in $S$, implying that the set $D_S$ is not a dominating set in $G_I$, contradicting the fact that $D_S$ is a semi-TD-set of $G_I$. Hence, there is no interval $I_h$ such that $b_0 < a_h$ and $b_h < a_{i_1}$, implying that $(0,i_1)$ is an (unmarked) arc in $D$.

Since $a_{n+1} > \max \{b_{k}\mid k \in [n]\}$, we note that every interval $I_h$ where $h \le n$ satisfies $b_h < a_{n+1}$. If there is an interval $I_h$ such that $h \le n$ and $b_{i_k} < a_h$, then the interval $I_h$ is not intersected by any interval in $S$, implying that the set $D_S$ is not a dominating set in $G_I$, a contradiction. Hence, there is no interval $I_h$ such that $b_{i_k} < a_h$ and $b_h < a_{n+1}$, implying that $(i_k,n+1)$ is an (unmarked) arc in $D$.

We show next that there is an arc $(i_j,i_{j+1})$ for every $j \in [k-1]$. Suppose, to the contrary, that the arc $(i_j,i_{j+1})$ does not exist for some $j \in [k-1]$. Thus, there is an interval $I_h$ such that $b_j < a_h$ and $b_h < a_{j+1}$. The interval $I_h$ is therefore not intersected by any interval in $S$, implying that the set $D_S$ is not a dominating set in $G_I$, a contradiction. Hence, there is no such interval $I_h$ such that $b_j < a_h$ and $b_h < a_{j+1}$, implying that either the intervals $I_{i_j}$ and $I_{i_{j+1}}$ overlap, in which case $(i_j,i_{j+1}) \in A_1$, or the intervals $I_{i_j}$ and $I_{i_{j+1}}$ do not overlap, in which case $(i_j,i_{j+1}) \in A_2$.

The above observations imply that there is a (directed) path $P \colon 0,i_1,i_2,\ldots,i_k,n+1$ between vertex~$0$ and vertex~$n+1$ in the digraph $D$. As observed earlier, the initial arc $(0,i_1)$ and the final arc $(i_k,n+1)$ of this path are unmarked arcs in $D$. It remains for us to show that no two consecutive arcs in the path $P$ are unmarked.

We note firstly that since the set $D_S$ is a semi-TD-set of $G_I$, the vertices $i_1$ and $i_2$ are within distance~$2$ apart in $G_I$, implying that either the intervals $I_{i_1}$ and $I_{i_{2}}$ overlap, in which case $(i_1,i_2) \in A_1$, or the intervals $I_{i_1}$ and $I_{i_{2}}$ do not overlap but there exists an interval $I_{h} \in I$ that overlaps both intervals $I_{i_1}$ and $I_{i_{2}}$, in which case $(i_1,i_2)$ is a marked arc in $A_2$. Analogously, the arc $(i_{k-1},i_{k}) \in A_1$ or $(i_{k-1},i_{k})$ is a marked arc in $A_2$. Hence, the second arc on the path $P$, namely the arc $(i_1,i_2)$ is not an unmarked arc, and the penultimate arc on the path $P$, namely the arc $(i_{k-1},i_{k})$ is not an unmarked arc.

We next consider two consecutive arcs $(i_j,i_{j+1})$ and $(i_{j+1},i_{j+2})$ on the path $P$ for some $j \in [k-2]$. Since the set $D_S$ is a semi-TD-set of $G_I$, the vertex $i_{j+1}$ is within distance~$2$ from at least one of the vertices $i_j$ and  $i_{j+2}$ in $G_I$. If $i_{j+1}$ is within distance~$2$ from the vertex $i_j$ in $G_I$, then the intervals $I_{i_j}$ and $I_{i_{j+1}}$ overlap, in which case $(i_j,i_{j+1}) \in A_1$, or the intervals $I_{i_j}$ and $I_{i_{j+1}}$ do not overlap but there exists an interval $I_{h} \in I$ that overlaps both intervals $I_{i_j}$ and $I_{i_{j+1}}$, in which case $(i_j,i_{j+1})$ is a marked arc in $A_2$. In particular, the arc $(i_j,i_{j+1})$ is not an unmarked arc. Analogously, if $i_{j+1}$ is within distance~$2$ from the vertex $i_{j+2}$ in $G_I$, then the arc $(i_{j+1},i_{j+2})$ is not an unmarked arc. The above observations imply that no two consecutive arcs in the path $P$ are unmarked. This completes the proof of Theorem~\ref{t:interval}.
\qed
\end{proof}

\medskip
By Lemma~\ref{l:interval} and Theorem~\ref{t:interval}, the problem of computing a minimum semi-TD-set in an interval graph is therefore equivalent to the problem of finding a shortest path in a directed graph with some sequencing constraints on certain arcs. However, we may reduce the latter problem to the ordinary shortest path problem in directed weighted graph. For this purpose, we construct another directed weighted graph $D'$ from $D$ in the following way:
\\[-24pt]
\begin{enumerate}
\item[$\bullet$] The vertex set of $D'$ is obtained by splitting each vertex $i$ of $D$ into two vertices $i_{\inR}$ and $i_{\outR}$, and so $|V(D')| = 2|V(D)|$.
\item[$\bullet$] The vertices $i_{\inR}$ and $i_{\outR}$ are joined by the arc $(i_{\inR},i_{\outR})$, whose length is zero.
\item[$\bullet$] For an arc in $A_1$ or a marked arc in $A_2$, say $(i,j)$, we add the arc $(i_{\outR},j_{\inR})$ with unit length in $D'$.
\item[$\bullet$] For an unmarked arc $(i,j)$ in $A_2$, we add the arc $(i_{\inR},j_{\outR})$ with unit length in $D'$.
\item[$\bullet$] For an arc $(0,i) \in A$, we add to $D'$ the arc $(0,i_{\outR})$ with zero length.
\item[$\bullet$] For an arc $(i,n+1) \in A$, we add the arc $(i_{\inR},n+1)$  with unit length in $D'$.
\end{enumerate}

Now, let $P$ be a shortest path from vertex $0$ to vertex $n+1$ in $D'$. Let $S = \{I_{i}\in I \mid$ either $i_{\inR}$ or $i_{\outR}$ or both belong to the path $P \, \}$. Then, the set of vertices in $G_I$ corresponding to the intervals in $S$ is a minimum semi-TD-set of $G_I$.

It can be seen that $D'$ can be constructed directly from $I$ in $O(n^{2})$ time. Since $D'$ is acyclic, a shortest path from vertex $0$ to vertex $n+1$ can also be computed in $O(n^{2})$ time. Hence, we are ready to state the main theorem of this section.

\begin{theorem}
\label{t:interval_main}
A minimum semi-TD-set in an interval graph with $n$ vertices can be computed in $O(n^{2})$ time.
\end{theorem}

\section{Approximation Results}\label{S:apx}

In this section, we establish upper and lower bounds on the approximation ratio of the \textsc{Minimum Semitotal Domination} problem. We also show that the \textsc{Minimum Semitotal Domination} problem is APX-complete even for bipartite graphs with maximum degree $4$.

\subsection{Approximation Algorithm}
In this subsection, we propose a $2+3\ln (\Delta+1)$-approximation algorithm for the \textsc{Minimum Semitotal Domination} problem. Our algorithm make use of two important graph optimization problems: (i) \textsc{Minimum Domination} problem, and (ii) \textsc{Minimum Set Cover} problem.

We first recall the definition of \textsc{Minimum Set Cover} problem.
Let $X$ be any non-empty set and $\mathcal{F}$ be
a collection of subsets of $X$. A set $\mathcal{C }\subseteq \mathcal{F}$ is called a cover of $X$, if every element of $ X$ belongs to at least one element of $\cC$. The
\textsc{Minimum Set Cover} problem is to find a minimum cardinality cover of $X$. The following approximation results are already known for the \textsc{Minimum Domination} problem and the \textsc{Minimum Set Cover} problem.

\begin{theorem}{\rm (\cite{cormen})}
\label{t:Approx_dom_set}
The \textsc{Minimum Domination} problem in a graph with maximum degree $\Delta$ can be approximated with an approximation ratio of $1 + \ln(\Delta + 1)$.
\end{theorem}

By Theorem~\ref{t:Approx_dom_set}, there exists an algorithm that outputs a dominating set, $D$, of a graph with maximum degree $\Delta$ in polynomial time and achieves the approximation ratio of $1 + \ln(\Delta+1)$; that is, $|D| \le (1 + \ln(\Delta + 1)) \gamma(G)$. Let APPROX-DOM-SET be such an approximation algorithm.

\begin{theorem}{\rm (\cite{cormen})}
\label{t:Min_set_cover}
The \textsc{Minimum Set Cover} problem for the instance $(X,\cF)$ can be approximated with an approximation ratio of $1 + \ln |S|$, where $S$ is a set of maximum cardinality in~$\cF$.
\end{theorem}

By Theorem~\ref{t:Min_set_cover}, there exists an algorithm that outputs a set cover, $C$, of instance $(X, \cF)$ in polynomial time and achieves the approximation ratio of $1 + \ln p$, where $p = \max \{|S| \colon S \in \cF\}$. Let APPROX-SET-COVER be such an approximation algorithm.

Next, we propose an algorithm APPROX-SEMI-TOTAL-DOM-SET to compute an approximate solution of the \textsc{Minimum Semitotal Domination} problem. Our algorithm works in two phases: In the first phase, we compute a dominating set $D$ of the given graph $G$ using algorithm APPROX-DOM-SET. In the second phase, we find an additional set $T$ of vertices such that $D \cup T$ becomes a semi-TD-set of $G$. We select the set $T$ in such a way that for every vertex  $v \in D$, there exists a vertex $w \in D \cup T$ such that $d(v,w) \le 2$. To select this set $T$, we first form an instance $(X,\cF)$ of the \textsc{Minimum Set Cover} problem, and then use the algorithm APPROX-SET-COVER to compute a set cover $\cC$ of $X$. Thereafter we construct $T$ from $\cC$ such that $D \cup T$ becomes a semi-TD-set of $G$.

Assume that $D$ is a dominating set of the given graph $G$ obtained by the algorithm APPROX-DOM-SET. Next, we illustrate the construction of an instance $(X,\cF)$ of the \textsc{Minimum Set Cover} problem. Recall that the distance $d_G(v,S)$ between a vertex $v$ and a set $S$ of vertices in $G$ is the minimum distance from~$v$ to a vertex of $S$ in $G$.
%Since $D$ is a dominating set of $G$, we note that every vertex of $D$ is within distance~$3$ from some other vertex of $D$.
Let $X$ be the subset of all vertices $v$ in $D$ such that $d_G(v,D \setminus \{v\}) \ge 3$.
%If $X \ne \emptyset$, let $X = \{x_{1},\ldots,x_{p}\}$.

Recall that for a vertex $v$ in the graph $G$, the set of all vertices within distance~$2$ from $v$ is denoted by $N_2[v]$; that is, $N_2[v] = \{v\} \cup N_1(v) \cup N_2(v)$. Let $V \setminus D = \{u_{1},\ldots,u_{q}\}$. Let $S_{j} = N_2[u_j] \cap X$ for $j \in [q]$. Thus, $S_j$ is the set of all vertices in $X$ at distance~$1$ or~$2$ from the vertex~$u_j$ in $G$. We now define $\cF = \{S_{1},\ldots,S_{q}\}$. Thus, $(X,\cF)$ forms an instance of the \textsc{Minimum Set Cover} problem. Now, let $\cC$ be the set cover obtained by the algorithm  APPROX-SET-COVER. Let $T=\{u_{j}\in V \setminus D \mid S_{j}\in \cC\}$. Note that for each vertex $x \in X$, there must exist at least one vertex $u_{j} \in T$ such that $x \in N_2(u_{j})$, because $\cC$ is a set cover of $X$. Thus, $D\cup T$ is a semi-TD-set of $G$. Next, we summarize our approximation algorithm in \textbf{Algorithm 1}.

\medskip
\begin{algorithm}[H]
\caption{\textbf{:} APPROX-SEMI-TOTAL-DOM-SET(G)}
 \textbf{Input:} A graph $G=(V,E)$.\\
\textbf{Output:} A semi-TD-set $D_{\st}$ of $G$.\\
\Begin{
$D=$APPROX-DOM-SET$(G)$;\\
Construct an instance $(X,\cF)$ of \textsc{Min Set Cover} problem;\\
 \eIf {$X=\emptyset$}{
 $D_{\st}=D$;}
{
 $\cC=$ SET-COVER$(X,\cF)$;\\
 $T=\{u_{j}\in V\setminus D \mid S_{j}\in \cC\}$;\\
 $D_{\st}=D\cup T$;

} return $D_{\st}$; }
\end{algorithm}

We note that the algorithm APPROX-SEMI-TOTAL-DOM-SET produces a semi-TD-set of a given graph $G$ in polynomial time. We are now in a position to prove the following theorem.

\begin{theorem} \label{th:stdapprox}
The \textsc{Minimum Semitotal Domination} problem in a graph $G$ with maximum degree $\Delta$ can be approximated with an approximation ratio of $2+3\ln(\Delta+1)$.
\end{theorem}
\begin{proof}
In order to prove the theorem, we show that the semi-TD-set, $D_{\st}$, produced by our algorithm APPROX-SEMI-TOTAL-DOM-SET is an approximate solution of the \textsc{Minimum Semitotal Domination} problem with an approximation ratio of $2+3\ln(\Delta+1)$; that is,
\[
|D_{\st}| \le (2+3\ln(\Delta+1)) \, \semiT(G).
\]

The algorithm APPROX-DOM-SET produces a dominating set, $D$, of $G$ with an approximation ratio $1 + \ln(\Delta + 1)$; that is,
\[
|D| \le (1 + \ln(\Delta + 1)) \gamma(G).
\]

Similarly, for the instance $(X,\cF)$ of the \textsc{Minimum Set Cover} problem, the algorithm APPROX-SET-COVER produces a cover $\cC$ of $X$ with an approximation ratio $(1+\ln p)$, where $p$ is the cardinality of largest set in $\cF$. Since the graph $G$ has maximum degree $\Delta$, we note that $|N_1(u_j)| \le \Delta$ and $|N_2(u_j)| \le \Delta(\Delta - 1)$ for each $j \in [q]$. Thus, $|S_{j}| \le \Delta^2$ for all $j \in [q]$. Thus, the cardinality of  each set of $\cF$ is at most~$\Delta^{2}$, implying that $p \le \Delta^{2}$. Thus if $\cC^{*}$ is an optimal cover of the instance $(X,\cF)$, then $|\cC|\le (1+\ln(\Delta)^{2}) \cdot |\cC^{*}|$, that is,
\[
|\cC|\le (1+2\ln(\Delta))\cdot |\cC^{*}|.
\]

Recall that $T = \{u_j \in V \setminus D \mid S_{j} \in C\}$, and so $|T| = |\cC|$. Let $T^{*}=\{u_{i}\in V \setminus D \mid S_{i}\in C^{*}\}$. Then, $|T^{*}|=|\cC^{*}|$, and $|T^{*}|$ denotes the minimum number of vertices needed to extend the dominating set $D$ of $G$ to a semi-TD-set of $G$. We note that the minimum number of vertices needed to extend a set of vertices to a semi-TD-set of $G$ is no more than the semitotal domination number of $G$. Hence, $|C^{*}| \le \semiT(G)$. By Observation~\ref{ob:chain}, $\gamma(G) \le \semiT(G)$. Hence,
\begin{eqnarray}
\nonumber |D_{\st}| = |D\cup T| &=& |D|+|T|\\
\nonumber &=& |D|+|\cC|\\
\nonumber &\le & (1+\ln(\Delta+1))\cdot \gamma(G) + (1+2\ln(\Delta))\cdot |\cC^{*}|\\
\nonumber &\le & (2+3\ln(\Delta+1)) \cdot \semiT(G).
\end{eqnarray}

The semi-TD-set, $D_{\st}$, produced by the algorithm APPROX-SEMI-TOTAL-DOM-SET is therefore an approximate solution of the \textsc{Minimum Semitotal Domination} problem with an approximation ratio of $2+3\ln(\Delta+1)$. This completes the proof of the Theorem~\ref{th:stdapprox}.~\qed
\end{proof}

\medskip
Since the \textsc{Minimum Semitotal Domination} problem in a graph with maximum degree~$\Delta$ admits an approximation algorithm that achieves the approximation ratio of $2+3\ln(\Delta+1)$, which is a poly-logarithmic function of $\Delta$, we immediately have the following corollary of Theorem~\ref{th:stdapprox}. %Recall that the class APX (an abbreviation of ``approximable") is the set of NP optimization problems that allow polynomial-time approximation algorithms with approximation ratio bounded by a constant.

\begin{cor}
The \textsc{Minimum Semitotal Domination} problem is in the class log-APX.
\end{cor}

\subsection{Lower bound on approximation ratio}
In this subsection, we establish a lower bound on the approximation ratio of the \textsc{Minimum Semitotal Domination} problem. To obtain the lower bound, we provide an approximation preserving reduction from the \textsc{Minimum Domination} problem. The following approximation harness result is already known for the \textsc{Minimum Domination} problem.

\begin{theorem}{\rm (\cite{chlebik})}
\label{t:dom-hard} For a graph $G = (V, E)$, the \textsc{Minimum Domination} problem cannot be approximated within $(1 - \epsilon)\ln |V| $ for
any $\epsilon > 0$ unless NP $\subseteq$  DTIME $(|V|^{O(\log  \log  |V|)})$.
\end{theorem}

Now, we are ready to prove the main theorem of this subsection.

\begin{theorem} \label{t:approx-hard}
 For a graph $G = (V, E)$, the \textsc{Minimum Semitotal Domination} problem cannot be approximated within
$(1 - \epsilon)\ln |V| $ for any $\epsilon > 0$ unless NP $\subseteq$  DTIME $(|V|^{O(\log  \log  |V|)})$.
\end{theorem}
\begin{proof}
We first describe an approximation preserving reduction from the \textsc{Minimum Domination} problem to the \textsc{Minimum Semitotal Domination} problem. This together with the inapproximability bound of the \textsc{Minimum Domination} problem will provide the desired result. Let $G=(V,E)$, where $V=\{v_{1},v_{2},\ldots,v_{n}\}$ be an arbitrary instance of the \textsc{Minimum Domination} problem. Now, we construct another graph $H=(V_{H},E_{H})$, an instance of the \textsc{Minimum Semitotal Domination} problem in the following way: $V_{H}=V\cup \{x_{i} \mid i \in [n] \} \cup \{y,z\}$ and $E_{H}=E\cup \{v_{i}x_{i},x_{i}y\mid i \in [n]\} \cup \{yz\}$. Note that $|V_{H}|=2|V|+2$.
The graph $G=(V,E)$, where $V=\{v_{1},v_{2},v_{3},v_{4}\}$ and $E=\{v_{1}v_{2},v_{2}v_{3},v_{3}v_{4}\}$, and the associated graph $H$ are shown in Fig.~\ref{f:hard} to illustrate the above construction.

\begin{figure}
 \begin{center}
  \includegraphics[width=8cm, height=3.5cm]{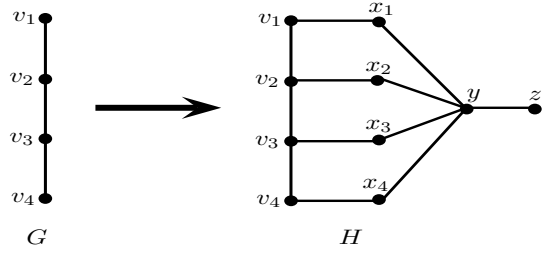}\\
 \caption{An illustration of the construction of $H$ from $G$  in the proof of Theorem~\ref{t:approx-hard}.}
\label{f:hard}
\end{center}
\end{figure}

Let $D^{*}$ denote a minimum dominating set of $G$ and $D_{st}^{*}$ a minimum semi-TD-set of $H$. Every dominating set of $G$ can be extended to a semi-TD-set of $H$ by adding to it the vertex $y$, implying that $|D_{st}^{*}|\le 1+|D^{*}|$.

Suppose that the \textsc{Minimum Semitotal Domination} problem can be approximated within a ratio of $\alpha$, where $\alpha=(1-\epsilon)\ln(|V_{H}|)$ for some fixed $\epsilon>0$, by using some algorithm, say Algorithm~A, that runs in polynomial time. Let $k$ be a fixed positive integer. Next, we propose an algorithm, say \textbf{ALGO-DOM-SET} to compute a dominating set of a given graph $G$ in polynomial time. Our algorithm is given in Algorithm~$2$.
\small
\begin{algorithm}[H]
\caption{\textbf{: ALGO-DOM-SET(G)}}
 \textbf{Input:} A graph $G=(V,E)$.\\
\textbf{Output:} A dominating set $D$ of $G$.\\
\Begin{
 \eIf {there exists a dominating set $S$ of cardinality at most $k$}{
 $D=S$;}
{
Construct the graph $H$;\\
Compute a semi-TD-set $D_{st}$ of $H$ using Algorithm~A;\\
$D=D_{st}\cap V$;\\
\For {i=1 to n } {
\If {$x_{i}\in D_{st}$} {
$D=D\cup \{v_{i}\}$;
}
}
}
return $D$;
}
\end{algorithm}

Note that  Algorithm~$2$ is a polynomial time algorithm as step $1$ runs in polynomial time and Algorithm~A is a polynomial time algorithm. If $|D^{*}|$ is at most $k$, then $D^{*}$ can be computed in polynomial time. Next we analyze the case when $|D^{*}|$ is greater than $k$.

If $D$ is a dominating set of $G$ produced by Algorithm~$2$, then $|D|\le |D_{st}|\le \alpha |D_{st}^{*}|\le \alpha (1+|D^{*}|)<\alpha (1+\frac{1}{k})|D^{*}|$. Therefore, Algorithm~$2$ approximates the \textsc{Minimum Domination} problem within ratio $\alpha (1+\frac{1}{k})$. Recall that $\alpha=(1-\epsilon)\ln(|V_{H}|)$ for some fixed $\epsilon>0$. Choosing the integer $k > 0$ large enough so that $1/k<\epsilon/2$, we note that
\[
\alpha \left(1+\frac{1}{k}\right) < \left(1-\epsilon\right) \left(1+ \frac{\epsilon}{2}\right) \ln(|V_{H}|) = (1-\epsilon')\ln(|V_{H}|) \approx (1-\epsilon')\ln(|V|)
\]
where $\epsilon' = \epsilon/2 + \epsilon^{2}/2$ and $|V_{H}|=2|V|+2$. This proves that the \textsc{Minimum Domination} problem can be approximated within ratio $(1-\epsilon')\ln(|V|)$ for some fixed $\epsilon'>0$. By Theorem \ref{t:dom-hard}, if the \textsc{Minimum Domination} problem can be approximated within  ratio $(1-\epsilon')\ln(|V|)$, then NP $\subseteq$  DTIME $(|V|^{O(\log  \log  |V|)})$. It follows that if the \textsc{Minimum Semitotal Domination} problem can be approximated within $(1-\epsilon)\ln(|V_{H}|)$ for any $\epsilon>0$, then  NP $\subseteq$  DTIME $(|V_{H}|^{O(\log  \log  |V_{H}|)})$. Hence, the \textsc{Minimum Semitotal Domination} problem for a graph $H=(V_{H},E_{H})$ cannot be approximated within $(1-\epsilon)\ln(|V_{H}|)$ for any $\epsilon>0$ unless  NP $\subseteq$  DTIME $(|V_{H}|^{O(\log  \log  |V_{H}|)})$.
\qed
\end{proof}

\subsection{APX-completeness}

By Theorem \ref{th:stdapprox},  the \textsc{Minimum Semitotal Domination} problem for bounded degree graphs can be approximated within a constant. Thus, the
\textsc{Minimum Semitotal Domination} problem for bounded degree graphs is in APX. In this subsection, we show that the \textsc{Minimum Semitotal Domination} problem is
APX-complete even for graphs with maximum degree~$4$. For this purpose, we recall the concept of L-reduction.

\begin{definition} Given two NP optimization problems $F$ and $G$ and a polynomial time transformation $f$ from instances of $F$ to instances of $G$, we say that $f$ is an L-reduction if there are positive constants $\alpha$ and $\beta$ such that for every instance $x$ of $F$ the following holds.
\begin{enumerate}
  \item $opt_{G}(f(x)) \le  \alpha \cdot opt_{F}(x)$.
  \item for every feasible solution $y$ of $f(x)$ with objective value $m_{G}(f(x),y)=c_{2}$
we can in polynomial time find a solution $y'$ of $x$ with
$m_{F}(x,y')=c_{1}$ such that $|opt_{F}(x)-c_{1}| \le \beta
|opt_{G}(f(x))-c_{2}|$.
\end{enumerate}

To show the APX-completeness of a problem $\Pi \in $APX, it suffices to show that there is an L-reduction from some APX-complete problem to $\Pi$.
\end{definition}

To show the APX-completeness of the \textsc{Minimum Semitotal Domination} problem, we give an L-reduction from the \textsc{Minimum Vertex Cover} problem.  A vertex \textsc{covers} an edge if it is incident with the edge. A \textsc{vertex cover} in a graph $G$ is a set of vertices that cover all the edges of $G$. A \textsc{minimum vertex cover} in $G$ is a vertex cover of $G$ of minimum cardinality. The \textsc{Minimum Vertex Cover} problem for a graph $G$ is to find a minimum vertex cover of~$G$.

\begin{theorem}{\rm (\cite{alimonti})}
\label{th:vc}
The \textsc{Minimum Vertex Cover} problem is APX-complete  for graphs with maximum degree~$3$.
\end{theorem}

Now, we are ready to prove the following theorem.

\begin{theorem}\label{t:apxcomplete}
The \textsc{Minimum Semitotal Domination} problem is APX-complete for graphs with maximum degree~$4$.
\end{theorem}
\begin{proof}
Since by Theorem~\ref{th:vc}, the \textsc{Minimum Vertex Cover} problem is APX-complete  for graphs with maximum degree~$3$, to complete the proof of the theorem, it is enough to establish an L-reduction $f$ from the instances of the \textsc{Minimum Vertex Cover} problem for graphs with maximum degree $3$ to the instances
of the \textsc{Minimum Semitotal Domination} problem for graphs with maximum degree $4$. Given a graph $G = (V, E)$, where $V = \{v_1, v_2, \ldots, v_n\}$ and $E = \{e_1, e_2,\ldots, e_m\}$, we construct
a graph $H = (V_H, E_H)$ in the following way.
\begin{enumerate}
\item For each vertex $v_{i}$, we add a cycle $C_{i} \colon u_{i}x_{i}y_{i}z_{i}u_{i}$ on four vertices to $H$, and we add the edge $u_{i}v_{i}$ joining the vertices $u_{i}$ and $v_{i}$ in $H$. Further, we add a new vertex $w_i$ and the pendant edge $u_{i}w_{i}$.
\item For each edge $e_{i}=(v_{j},v_{k})$ in $E$, introduce a vertex $e^{i}$ in $H$. Also add the edges $e^{i}v_{j}$ and  $e^{i}v_{k}$ in $H$.
\end{enumerate}

Formally, the vertex set $V_{H}=V\cup \{u_{i},x_{i},y_{i},z_{i},w_{i}\mid i \in [n] \} \cup \{e^{i}\mid e_{i}\in E\}$ and edge set $E_{H}=\{v_{i}u_{i},u_{i}w_{i},u_{i}x_{i},x_{i}y_{i},y_{i}z_{i},z_{i}u_{i} \mid i \in [n]\} \cup \{v_{j}e^{k} \mid  v_j$ is incident to $e_k$ in $G\}$. Note that the maximum degree of $H$ is~$4$. The construction of a graph $H$ from a graph $G$ is illustrated in Fig.~\ref{f:apx}. Now, we first prove the following claim.

\begin{figure}
 \begin{center}
  \includegraphics[width=11cm, height=5.5cm]{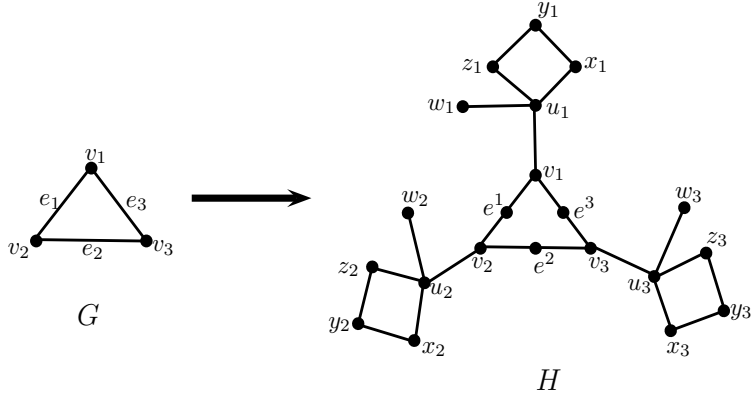}\\
 \caption{An illustration of the construction of $H$ from $G$  in the proof of Theorem~\ref{t:apxcomplete}.}
\label{f:apx}
\end{center}
\end{figure}

\begin{claim}
\label{c:claim_VC}
If $V_{c}^{*}$ is a minimum vertex cover of $G$, then $|V_{c}^{*}| = \semiT(H) - 2n$. %and $D_{st}^{*}$ is a minimum semi-TD-set of $H$, then $|D_{st}^{*}| = |V_{c}^{*}|+2n$.
\end{claim}
\begin{proof}
Let $V_{c}^{*}$ be a minimum vertex cover of $G$. The set $V_{c}^{*} \cup \{u_{i},y_{i}\mid i \in [n]\}$ is a semi-TD-set of $H$, and so $\semiT(H) \le |V_{c}^{*}| + 2n$.
Now, let $D_{st}^{*}$ be a minimum semi-TD-set of $H$, and so $\semiT(H) = |D_{st}^{*}|$. For each $i \in [n]$, in order to dominate the vertices $w_i$ and $y_i$, we note that $|D_{st}^{*}\cap \{u_{i},w_{i},x_{i},y_{i},z_{i}\}|\ge 2$. Also, for each $j \in [m]$, if $e_j = v_rv_s$ is an arbitrary edge in $G$, then in order to dominate the vertex $e^{j}$ in $H$ we note that $e^{j} \in D_{st}^{*}$ or $D_{st}^{*}$ contains at least one of $v_{r}$ or $v_{s}$. If $e^{j}$ belongs to $D_{st}^{*}$, then replace $e^{j}$ with either $v_{r}$ or $v_{s}$ in $D_{st}^{*}$. Now define $V_{c} = D_{st}^{*}\cap V$. The resulting set $V_{c}$ contains at least one end of the edge $e_{j}$ in $G$ for all $j \in [m]$, implying that $V_{c}$ is a vertex cover of $G$. Therefore, $|V_{c}^{*}| \le |V_{c}| \le |D_{st}^{*}|-2n = \semiT(H) - 2n$. As observed earlier, $|V_{c}^{*}| \ge \semiT(H) - 2n$. Consequently, $|V_{c}^{*}| = \semiT(H) - 2n$. This completes the proof of the claim.
\qed
\end{proof}

We now return to the proof of Theorem~\ref{t:apxcomplete}.  Let $V_{c}^{*}$ is a minimum vertex cover of $G$, and let $D_{st}^{*}$ be a minimum semi-TD-set of $H$. By Claim~\ref{c:claim_VC}, $\semiT(H) = |D_{st}^{*}| = |V_{c}^{*}| + 2n$. Since the maximum degree of $G$ is~$3$, we note that $n \le 3|V_{c}^{*}|$ and therefore $|D_{st}^{*}| \le 7|V_{c}^{*}|$.

Now, consider a semi-TD-set, say $D_{st}$, of $H$. For each $j \in [m]$, if $e_j = v_rv_s$ is an edge in $G$ and $e^{j}$ belongs to $D_{st}$, then we replace $e^{j}$ in $D_{st}$ with either $v_{r}$ or $v_{s}$. Let $V_{c} = D_{st} \cap V$. Analogously as in the proof of Claim~\ref{c:claim_VC}, the set $V_{c}$ is a vertex cover of $G$ and $|V_{c}| \le |D_{st}| - 2n$. Hence, $|D_{st}| - |D_{st}^{*}| \ge |V_{c}| + 2n - |D_{st}^{*}| = |V_{c}|-|V_{c}^{*}|$.  This proves that $f$ is an L-reduction with $\alpha=7$ and $\beta=1$, and completes the proof of Theorem~\ref{t:apxcomplete}.
\qed
\end{proof}

It can be observed that the graph $H$ constructed in Theorem~\ref{t:apxcomplete} is also a bipartite graph. Hence as an immediate corollary of Theorem~\ref{t:apxcomplete}, we have the following result.

\begin{cor}
The \textsc{Minimum Semitotal Domination} problem is APX-complete for bipartite graphs with maximum degree~$4$.
\end{cor}

\section{Conclusion}
\label{S:conclude}

In this paper, we have resolved the complexity status of the problem for chordal bipartite graphs, chordal graphs, planar graphs and interval graphs. We have proved that the \textsc{Semitotal Domination Decision} problem remains NP-complete for chordal graphs. We also proved that the \textsc{Minimum Semitotal Domination} problem is polynomial time solvable for interval graphs. It will be interesting to look for the complexity status of the problem for strongly chordal graphs, which is a subclass of chordal graphs, but superclass of interval graphs. One may also look for the complexity status of the problem for subclasses of chordal bipartite graphs, such as bipartite permutation graphs and convex bipartite graphs. We have also studied the approximation aspects of the problem. We have proved that the \textsc{Minimum Semitotal Domination} problem remains APX-complete even for bipartite graphs with maximum degree $4$. As the \textsc{Minimum Semitotal Domination} problem is easily solvable for graphs with maximum degree $2$, it will also be interesting to look for the complexity status of the problem for graphs with maximum degree $3$.

\end{document}